\documentclass[mathematics,article,accept,pdftex,oneauthor]{mdpi_template/Definitions/mdpi}

\firstpage{1}
\pubvolume{1}
\issuenum{1}
\articlenumber{0}
\pubyear{2026}
\copyrightyear{2026}
\datereceived{ }
\daterevised{ }
\dateaccepted{ }
\datepublished{ }

\makeatletter
\let\@origmaketitle\@maketitle
\renewcommand{\@maketitle}{%
  \let\origincludegraphics\includegraphics
  \renewcommand{\includegraphics}[2][]{%
    \def\@tempa{##2}%
    \@expandtwoargs\in@{.eps}{\@tempa}%
    \ifin@ \else \origincludegraphics[##1]{##2}\fi
  }%
  \@origmaketitle
  \let\includegraphics\origincludegraphics
}
\fancypagestyle{plain}{%
  \fancyhf{}%
  \fancyfoot[C]{\footnotesize\thepage}%
}
\renewcommand{\cright}{}
\renewcommand{\leftcolDates}{}
\renewcommand{\leftcolumnUpdatelogo}{}
\renewcommand{\leftcolumnCorrstatement}{}
\renewcommand{\leftcolumnCright}{}
\makeatother

\newgeometry{left=2.5cm, right=2.5cm, top=1.27cm, bottom=1.08cm,
  includehead, includefoot, footskip=36pt, headsep=24pt}
\fancyheadoffset{0pt}
\fancyfootoffset{0pt}
\usepackage{algorithm}
\usepackage{algpseudocode}
\usepackage{siunitx}
\usepackage{placeins}
\DeclareMathOperator{\erfcx}{erfcx}
\DeclareMathOperator{\erfc}{erfc}

\Title{Implying Volatility: How Fast Can We Go?}

\Author{Fabien Le Floc'h\orcidA{} and Jherek Healy}

\AuthorNames{Fabien Le Floc'h and Jherek Healy}

\address{%
Independent researcher; fabien@2ipi.com}

\corres{Correspondence: fabien@2ipi.com}

\abstract{%
FlashIV is a low-latency Black--Scholes implied-volatility solver for production
use.  It normalises each input to an out-of-the-money price and solves a
tail-stable erfcx/log-price residual.  The hot path combines a cheap
Li/asymptotic seed with a fixed, branch-light Householder refinement and guarded
boundary handling.  Across regular and stressed benchmarks, FlashIV stays close
to multiprecision Black reference prices while reducing latency in the reported
benchmark.  Detailed comparisons include recent implied-volatility solvers and
J\"ackel's \emph{Let's Be Rational}.  An optional guarded Newton correction gives
tighter agreement with a J\"ackel-style reference price when applications need
reference-price alignment.}

\keyword{implied volatility; Black--Scholes; Householder iteration; Chebyshev method;
erfcx; root-finding; computational finance}

\begin{document}

%%%%%%%%%%%%%%%%%%%%%%%%%%%%%%%%%%%%%%%%%%
\section{Introduction}
\label{sec:intro}

Implied volatility inversion, recovering the Black--Scholes volatility
$\sigma$ from an observed option price, is among the most frequently
executed numerical tasks in quantitative finance.  A single calibration
of a stochastic volatility surface may require $10^4$--$10^6$ such
inversions; real-time risk systems demand throughput of millions of
inversions per second.

The problem has received sustained attention.  J\"ackel~\cite{jaeckel2017}
introduced a solver of remarkable robustness, using region-dependent
asymptotic expansions, a complementary objective for the upper branch,
and carefully crafted log-space iterations.
Li~\cite{li2006} proposed a bivariate rational approximation as an
initial guess, later refined by Stefanica and
Radoi\v{c}i\'{c}~\cite{sr2017} (hereafter SR) who gave a closed-form
global approximation based on P\'{o}lya's bound for the cumulative
normal distribution, with
uniform relative-error guarantees for the implied-volatility guess.
Choi et~al.~\cite{choi2023} derived tighter implied volatility bounds
from option-delta inequalities and used those bounds to prove monotone
convergence of a Newton--Raphson root finder.
Healy~\cite{healy2024} combined the SR initial guess with
Householder iteration on the log price, providing a practical
implementation in his treatment of equity derivative pricing.

The solver we present, FlashIV (Fast Li/asymptotic seeded Householder implied
volatility), follows the standard two-stage structure, initial
guess followed by iterative refinement, but applies a sequence of
targeted optimisations to each stage.  Our techniques target
these bottlenecks:
\begin{itemize}
    \item A hybrid initial guess that replaces the SR formula with
  an approximately 4.1\,ns asymptotic approximation for out-of-the-money
      options, protected by near-ATM and high-price complementary guards
      (Section~\ref{sec:guess}).
  \item A two-tier erfcx strategy with a cheap pre-step and
        a full-precision polish (Section~\ref{sec:erfcx}).
  \item Elimination of the convergence loop in favour of a fixed
        iteration count with a conditional safety net
        (Section~\ref{sec:fixed}).
\end{itemize}

%%%%%%%%%%%%%%%%%%%%%%%%%%%%%%%%%%%%%%%%%%
\section{Problem Formulation}
\label{sec:problem}

\subsection{Normalization}
\label{sec:normalization}

All prices are first reduced to the same normalised OTM representation.
Let $C$ and $P$ denote undiscounted call and put prices with forward
$F$, strike $K$, and time to expiry~$T$.  We call an option at the money
(ATM) when $F=K$, and near ATM when $|\ln\left(F/K\right)|$ is small.  A call is
out of the money (OTM) when $F<K$, while a put is OTM when $F>K$; the
opposite cases are in the money (ITM).  Rather than carrying separate
call, put, ITM, and OTM formulae through the solver, we order the
forward--strike pair as
\[
  F_* = \min(F,K), \qquad K_* = \max(F,K), \qquad
  \frac{F_*}{K_*}=\min(F/K,K/F) \le 1.
\]
The corresponding OTM value is obtained by put--call parity, and by
exchanging the forward and strike for OTM puts:
\[
C_{\mathrm{OTM}}=
\begin{cases}
  C, & \text{call and } F\le K,\\
  C-(F-K), & \text{call and } F>K,\\
  P, & \text{put and } F>K,\\
  P-(K-F), & \text{put and } F<K.
\end{cases}
\]
Thus every admissible input is represented as an undiscounted OTM call
on $(F_*,K_*)$.  The variables passed to the scalar inversion are
\begin{linenomath}
\begin{equation}\label{eq:normalise}
  x = \ln\left(\frac{F_*}{K_*}\right) \le 0, \qquad
  e^x = \frac{F_*}{K_*}, \qquad
  c = \frac{C_{\mathrm{OTM}}}{F_*}, \qquad
  v = \sigma\sqrt{T}.
\end{equation}
\end{linenomath}
Here $v$ denotes total volatility.  When comparing with J\"ackel's
sqrt-forward normalisation, we also use
$\beta=c e^{x/2}=C_{\mathrm{OTM}}/\sqrt{F_*K_*}$.

This normalisation removes intrinsic value before inversion.  In
particular, originally ITM options are priced through their OTM parity
legs, avoiding the cancellation that would occur if one subtracted a
large intrinsic component from the original option price.

In these coordinates the Black formula for the normalised OTM call price is
\begin{linenomath}
\begin{equation}\label{eq:black}
  c(x, v) = \Phi\left(\frac{x}{v} + \frac{v}{2}\right)
    - e^{-x}\Phi\left(\frac{x}{v} - \frac{v}{2}\right),
\end{equation}
\end{linenomath}
where $\Phi$ denotes the standard normal cumulative
distribution function.

\subsection{Log-Price Objective and erfcx Decomposition}
\label{sec:log-price-objective}
Following J\"ackel~\cite{jaeckel2017} and Choi et~al.~\cite{choi2023}, the iteration is performed in log-price
space rather than on the raw price residual.  For a target normalised
price $c_{\mathrm{target}}$, the objective is
\begin{linenomath}
\begin{equation}\label{eq:objective}
  f(v) = \ln\left(c(x,v)\right) - \ln\left(c_{\mathrm{target}}\right).
\end{equation}
\end{linenomath}

\begin{Proposition}\label{prop:logc}
For $h = x/v$ and $t = v/2$, the log-price admits the
representation
\begin{equation}\label{eq:logc}
  \ln\left(c\right) = -\tfrac{1}{2}(h^2 + t^2) - \ln\left(2\right) - \tfrac{x}{2}
           + \ln\left(N^+ - N^-\right),
\end{equation}
where
\begin{equation}\label{eq:erfcx}
  N^+ = \erfcx\left(-(h+t)/\sqrt{2}\right), \qquad
  N^- = \erfcx\left(-(h-t)/\sqrt{2}\right).
\end{equation}
\end{Proposition}

\begin{proof}
Using~\eqref{eq:black} and
$\Phi\left(z\right) = \tfrac{1}{2}\erfc\left(-z/\sqrt{2}\right)$ gives
\[
  c = \tfrac{1}{2}\bigl[
    \erfc\left(-(h+t)/\sqrt{2}\right)
    - e^{-x}\,\erfc\left(-(h-t)/\sqrt{2}\right)
  \bigr].
\]
Now write $\erfc\left(z\right) = e^{-z^2}\erfcx\left(z\right)$.  Since $ht=x/2$,
\[
  e^{-(h+t)^2/2}=e^{-(h^2+t^2)/2-x/2},
  \qquad
  e^{-x}e^{-(h-t)^2/2}=e^{-(h^2+t^2)/2-x/2}.
\]
The two terms therefore share the common factor
$e^{-(h^2+t^2)/2-x/2}/2$.  Factoring it out and taking logarithms
gives~\eqref{eq:logc}.
\end{proof}

\begin{Remark}
The scaled complementary error function
$\erfcx\left(z\right) = e^{z^2}\erfc\left(z\right)$ is bounded and positive for all
real~$z$.  Consequently, the difference $N^+ - N^-$ remains computable
even when the price itself would underflow to zero in double precision.
This is the main reason for using the log-price formulation: the
objective~$f(v)$ remains well defined and smoothly differentiable for
arbitrarily deep OTM options.
\end{Remark}

%%%%%%%%%%%%%%%%%%%%%%%%%%%%%%%%%%%%%%%%%%
\section{Solver Design}
\label{sec:solver}

\subsection{Householder Refinement of Order~3}
\label{sec:householder}

We use Householder's method of order~3 (H3), which achieves
\emph{quartic} local convergence at a simple root~\cite{householder1970}.
For the log-price objective the admissible root is simple because
$f'(v)=\phi\left(d_1\right)/c(x,v)>0$.  We use this only as a local refinement
property: once the seed is in the H3 basin, two exact H3 steps are more
than enough for double precision in the tested regimes.  For FlashIV, the fast
pre-step moves the inexpensive Li/asymptotic seed into the local H3 basin.  With
$\eta_n = -f(v_n)/f'(v_n)$ denoting the Newton displacement, the update is
\begin{linenomath}
\begin{equation}\label{eq:h3}
  v_{n+1} = v_n + \eta_n \cdot
  \frac{1 + \tfrac{1}{2}\,\delta_2\,\eta_n}
       {1 + \delta_2\,\eta_n + \tfrac{1}{6}\,\delta_3\,\eta_n^{2}},
\end{equation}
\end{linenomath}
where $\delta_2 = f''/f'$ and $\delta_3 = f'''/f'$ are the
normalised second and third derivatives of the objective.

\begin{Proposition}\label{prop:derivatives}
Let $\ell(v)=\ln(c(x,v))$.  Since
$f(v)=\ell(v)-\ln(c_{\mathrm{target}})$, the derivative ratios used by H3 are
those of~$\ell$.  They admit the closed-form expressions
\begin{align}
  \ell'(v) &= \frac{2/\sqrt{2\pi}}{N^+ - N^-},
  \label{eq:logvega}\\[3pt]
  \frac{\ell''(v)}{\ell'(v)} &= \frac{(h+t)(h-t)}{v} - \ell'(v),
  \label{eq:d2}\\[3pt]
  \frac{\ell'''(v)}{\ell'(v)} &=
    \frac{-3h^2 - t^2 + (h^2-t^2)^2}{v^2}
    - 3\,\ell'(v)\,\frac{\ell''(v)}{\ell'(v)}
    - \bigl[\ell'(v)\bigr]^2.
  \label{eq:d3}
\end{align}
\end{Proposition}

\begin{proof}
The vega of the normalised Black formula is
$\partial c / \partial v = \phi\left(d_1\right)$, where
$\phi$ is the standard normal density.  Thus
$\ell'(v) = c^{-1} \phi\left(d_1\right)$.  Expressing
$\phi\left(d_1\right) = (2\pi)^{-1/2} e^{-d_1^2/2}$ and using the
erfcx representation~\eqref{eq:erfcx}, a direct computation
gives~\eqref{eq:logvega}.  The ratios~\eqref{eq:d2}
and~\eqref{eq:d3} follow by differentiating $\ell'(v)$
with respect to~$v$ and simplifying, using
$dh/dv = -h/v$ and $dt/dv = 1/2$.
\end{proof}

\begin{Remark}\label{rem:no-extra-erfcx}
Equations~\eqref{eq:d2}--\eqref{eq:d3} involve only $\ell'(v)$,
$h$, $t$, and~$v$, with no additional erfcx evaluations.  Once $N^+$
and $N^-$ are known from the objective evaluation, all three
derivatives follow from elementary arithmetic ($\sim$25 mul-adds).
This makes H3 only marginally more expensive than Halley
(order~2, cubic convergence) while converging one order faster.
\end{Remark}

\begin{Remark}[Why not Newton?]
Newton's method (order~1, quadratic convergence) requires
$\sim$4 iterations from a typical initial guess, totalling
$4 \times 2 \times 5.5 \approx 44$\,ns of erfcx evaluation
alone.  Halley (order~2, cubic) reduces this to $\sim$2.07
average iterations but occasionally needs a third ($\sim$7\% of
cases).  H3 at the same per-step erfcx cost achieves $\sim$2.01
average iterations with a maximum of~3 and, crucially, enables
the fixed-iteration strategy of Section~\ref{sec:fixed}.
\end{Remark}

\begin{Remark}[Why not higher orders?]
Householder-4 (order~4, quintic convergence) would require
$f^{(iv)}$, adding $\sim$10 mul-adds per step, but saving at most
0.01 iterations on average.  At $\sim$0.3\,ns of extra arithmetic
versus $\sim$0.3\,ns of amortised erfcx savings, the tradeoff is
a wash.
\end{Remark}

\subsection{The Two-Tier erfcx Strategy}
\label{sec:erfcx}

Each H3 step requires two erfcx evaluations (for $N^+$ and~$N^-$).
A full-precision erfcx costs $\sim$6.2\,ns per call (Boost rational
approximation, as implemented in Apache Commons Numbers~1.2) and
delivers 15+ significant digits.  The first step of the solver,
however, does not need 15~digits: its purpose is merely to promote an
initial guess from $\sim$2--3 digits of accuracy to $\sim$6--8
digits, after which the exact polish takes over.

\begin{Definition}[Pre-step erfcx]\label{def:fast-erfcx}
The pre-step erfcx is a composite of two classical approximations:
\begin{enumerate}
  \item For $z \ge 2.5$: the asymptotic expansion
        $\erfcx\left(z\right) \approx
        \frac{1}{\sqrt{\pi}\,z}\bigl(1 - \tfrac{1}{2z^2}
        + \tfrac{3}{4z^4} - \tfrac{15}{8z^6}\bigr)$.
  \item For $0 \le z < 2.5$: the Abramowitz \& Stegun~\cite{abramowitz}
    rational polynomial (formula 7.1.26).  In this form
    $\erfc\left(z\right) \approx p(t)e^{-z^2}$ with
    $t = 1/(1 + 0.3275911z)$, so $\erfcx\left(z\right)$ is evaluated as
    the polynomial $p(t)$ directly.
  \item For $z < 0$: the reflection identity
        $\erfcx\left(z\right) = 2e^{z^2} - \erfcx\left(-z\right)$.
\end{enumerate}
The total cost is $\sim$2.3\,ns per call.
\end{Definition}

\begin{Remark}
The worst-case relative error of the A\&S approximation is
$2.8 \times 10^{-3}$ at $z \approx 2.5$, the junction between the
rational and asymptotic pieces.  Although this is only 2.5 significant
digits, it is more than adequate for the pre-step: the purpose is to
reduce the residual from $O(1)$ to $O(10^{-6})$, not to achieve
machine precision.
\end{Remark}

We also tested a degree-12 Chebyshev polynomial erfcx on $[1.5, 5.0]$
achieving $\sim$9~significant digits ($\sim$3.8\,ns per call).
Despite reducing the average polish iteration count from 2.07
to 2.01, the $\sim$2\,ns extra pre-step cost offsets the
$\sim$2\,ns saved from fewer polish iterations.  The A\&S
approximation sits at the sweet spot of the accuracy--cost tradeoff.

\subsection{The Initial Guess}
\label{sec:guess}

The choice of initial guess has a large impact on total solver time,
not through iteration count (which varies by at most~1 between all
reasonable guesses) but through the \emph{direct cost of computing the
guess itself}.

\subsubsection{Li's rational approximation (2006)}
A degree-$(3,3)$ bivariate rational polynomial~\cite{li2006}:
\begin{linenomath}
\begin{equation}\label{eq:li}
  v_{\text{Li}}(x, c) = \frac{\sum_{i+j \le 3} m_{ij}\, x^i\, c^j}
                              {\sum_{i+j \le 3} n_{ij}\, x^i\, c^j},
\end{equation}
\end{linenomath}
with 20 pre-fitted coefficients.  Valid on the domain
$\mathcal{D}_{\text{Li}} = \{(x, c) : |x| < 3,\;
0.0005 < c < 0.9995\}$.  Cost: $\sim$3.3\,ns.

\subsubsection{The asymptotic OTM guess}
\label{sec:asym}
For deep OTM options ($c \le 0.5$, outside $\mathcal{D}_{\text{Li}}$),
we derive a simple closed-form initial guess from the tail asymptotics
of the Black formula.

\begin{Proposition}\label{prop:asymguess}
For small $c$, the normalised Black call price satisfies
$\ln\left(c\right) \approx -\tfrac{1}{2}\,d_1^2 - \tfrac{1}{2}\ln\left(2\pi\right)$
to leading order.  Defining
$D = \sqrt{-2\ln\left(c\right) - \ln\left(2\pi\right)}$
as an approximation to~$|d_1|$ and solving the quadratic
$d_1 = x/v + v/2$ for~$v$ in rationalised form yields
\begin{equation}\label{eq:asymguess}
  v_0 = \frac{-2x}{D + \sqrt{D^2 - 2x}}.
\end{equation}
\end{Proposition}

\begin{proof}
For small $c$, the Black formula~\eqref{eq:black} is dominated
by the first term: $c \approx \Phi\left(d_1\right) \approx
\phi\left(d_1\right) / |d_1|$.  Taking logarithms,
$\ln\left(c\right) \approx -\tfrac{1}{2} d_1^2
  - \tfrac{1}{2}\ln\left(2\pi\right) + \ln\left(1/|d_1|\right)$.
Dropping the $\ln\left(1/|d_1|\right)$ correction (which is $O(\ln\left(\ln\left(1/c\right)\right))$
relative to the $O(\ln\left(1/c\right))$ leading term) gives
$d_1 \approx -D$.  The quadratic $v^2/2 + d_1 v + x = 0$ (from
$d_1 = x/v + v/2$) has the positive root
$v = -d_1 + \sqrt{d_1^2 - 2x}$.  Substituting $d_1 \approx -D$
and rationalising (multiplying numerator and denominator by
$D - \sqrt{D^2 - 2x}$) yields~\eqref{eq:asymguess}, which avoids
catastrophic cancellation when $|x| \ll D^2$.
\end{proof}

\begin{Remark}
The cost of~\eqref{eq:asymguess} is $\sim$3.4\,ns: one subtraction,
two square roots, and one division, with $\ln\left(c\right)$ already available
from the log-space objective setup.  The approximation is 2--3 digits
accurate for moderate OTM and improves monotonically as $c \to 0$.
\end{Remark}

\begin{Remark}[Validity guard]\label{rem:guard}
The leading-order approximation degrades when $c$ is not small
enough for the tail asymptotics to dominate.  We impose the guard
$\ln\left(c\right) < -2$ (i.e., $c < e^{-2} \approx 0.135$) in addition to
$c \le 0.5$, preventing activation for high-volatility, long-maturity
cases where $c \approx 0.5$ despite deep OTM ($|x| > 3$).
\end{Remark}

\subsubsection{The hybrid Li+asymptotic dispatch}
The solver-side dispatch around the Li/asymptotic guess is as follows.
In the floating-point implementation the high-price trigger is
$c_g=\operatorname{nextDown}(0.99)$, i.e., the representable double
immediately below 0.99; analytically this is the same upper-polish
guard, but it avoids sending a one-ulp-rounded boundary price back to
the direct log-price region.
\begin{enumerate}
    \item[(i)] $c \le 10^{-6}$, $|x| \le 10^{-8}$, and $x \le 0$:
      Bachelier-limit branch.  In the sqrt-forward normalisation
      $\beta=c e^{x/2}$ with log-moneyness gap $m=-x$, the Black price
      is expanded around the Bachelier integral.  The deep tail
      ($\ln\left(m/\beta\right) > 20$ and $m/v > 4$) is handled by a scaled
      Mills-ratio solver in the normal model, returning the volatility
      directly.  Otherwise, the Le Floc'h rational Bachelier
      approximation~\cite{lefloc2016basispoint} provides a
      Bachelier-normal seed accurate to $\sim 10^{-14}$, polished by
      two Newton corrections on the local Black--Bachelier expansion.
      If both guarded paths decline, the branch returns the zero-volatility
      limit as a defensive terminal value; this path was not reached in the
      validation grids because the deep-tail solver covers ratios below the
      rational approximation's finite domain.  The guarded branch is tested before
      the general initial-guess dispatch.  The full derivation of this shared
      microscopic branch is given in the ThiopheneIV paper~\cite[Appendix~A.3]{lefloc2026thiophene}.
  \item[(ii)] $|x| < 3$ and $0.0005 < c < c_g$: Li rational
      ($\sim$4\,ns).
  \item[(iii)] $c \le 0.0005$ and $|x| < 0.01$: a near-ATM small-price
      seed $v_0 = \sqrt{x^2 + 2\pi c^2}$, which reduces to the
      first-order ATM relation $c \approx v/\sqrt{2\pi}$.
    \item[(iv)] $c \ge c_g$: the current seed is lifted by an upper-price
      asymptotic seed based on $q=1-c$, followed by two Householder steps on the
      complementary $\ln\left(q\right)$ objective.
    \item[(v)] $c \le 0.5$ and $\ln\left(c\right) < -2$: asymptotic OTM
      formula~\eqref{eq:asymguess} ($\sim$4.1\,ns).
  \item[(vi)] Otherwise: put-side asymptotic or $\sqrt{2|x|}$ fallback
      ($< 1$\,ns; rare).
\end{enumerate}

Table~\ref{tab:guessdomain} reports the guess-domain distribution
across all eight benchmark datasets for cases that reach the general
initial-guess dispatch.  The Bachelier-limit branch is part of the
high-level dispatch, but it is a terminal microscopic-price branch
rather than an initial-guess domain.

\begin{table}[!htbp]
\caption{Guess-domain distribution by dataset.\label{tab:guessdomain}}
\begin{tabularx}{\textwidth}{lCCCC}
\toprule
\textbf{Dataset} & \textbf{Cases} & \textbf{Li (\%)} &
\textbf{Asym (\%)} & \textbf{Fallback (\%)} \\
\midrule
CLY-3D    & \num{51321} & 58.5 & 41.5 & 0.0 \\
CLY-20    & \num{1600}  & 76.0 & 24.0 & 0.0 \\
CLY-80    & \num{1600}  & 84.1 & 15.9 & 0.0 \\
J\"ackel  & \num{5182}  & 72.8 & 27.2 & 0.0 \\
Market    & \num{7151}  & 80.3 & 19.7 & 0.0 \\
Corners   & 278         & 22.3 & 74.8 & 2.9 \\
Stress    & \num{1270}  & 50.6 & 49.1 & 0.4 \\
HighVol   & 149         &  0.0 & 49.7 & 50.3 \\
\bottomrule
\end{tabularx}
\end{table}

On the CLY-3D benchmark, the Li rational guess covers 58.5\% of cases
and the asymptotic OTM guess covers 41.5\%, with the $\sqrt{2|x|}$
fallback never invoked.  The measured average Li/asymptotic guess cost is
$\sim$4.8\,ns.

Figure~\ref{fig:branch-domain} shows the corresponding dispatch geometry
in normalised OTM coordinates.  The main panel uses a logit scale for $c$,
with tick labels shown as prices, so that the high-price guard at $c\ge0.99$
is visible rather than compressed against the top edge.  The microscopic
Bachelier-limit guard remains too narrow to see at that scale, so it is
resolved in the right-hand zoom panel.

\begin{figure}[!htbp]
\centering
\includegraphics[width=\textwidth]{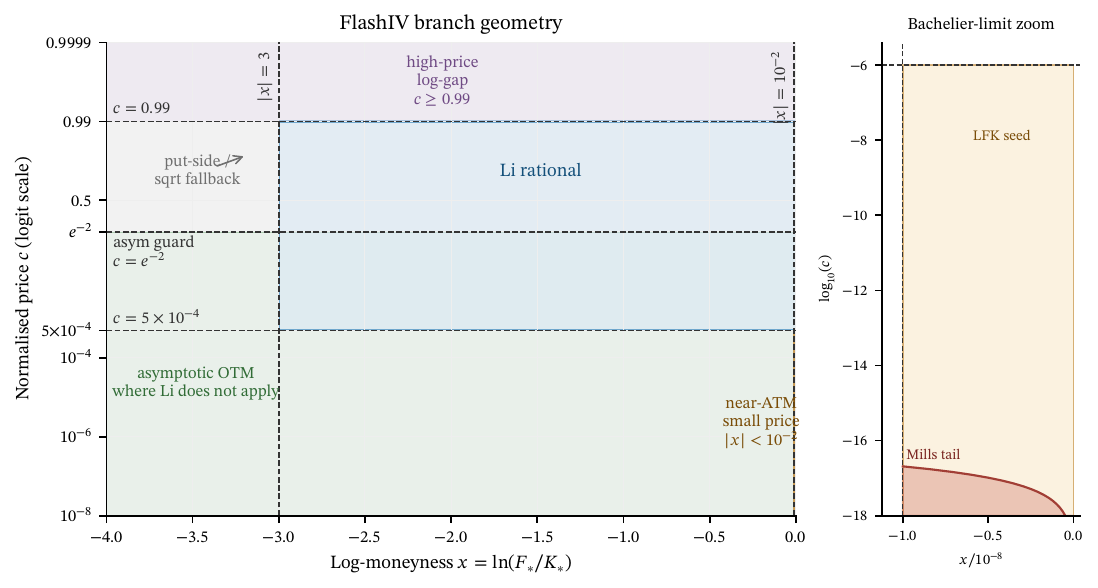}
\caption{FlashIV branch geometry in normalised OTM coordinates.  The vertical
axis uses a logit scale for the normalised OTM price $c$.  Pale
regions show the high-price complementary log-gap, Li rational, near-ATM small-price, and
asymptotic OTM regimes; the pale gray remainder is the inexpensive
put-side-asymptotic or $\sqrt{2|x|}$ fallback region.  The left gray zone is
therefore not the Bachelier-limit branch.  The $|x|=10^{-2}$ line is the
activation guard for the near-ATM small-price seed.  The right panel resolves the
microscopic Bachelier-limit box,
where the Mills-ratio tail treatment covers the deepest prices and the
rational Bachelier seed covers the transition region.  The $c=e^{-2}$ line is
the validity guard for activating the asymptotic OTM seed outside the Li domain.\label{fig:branch-domain}}
\end{figure}

\subsubsection{StrontiumIV: guarded Stefanica---Radoi\v{c}i\'{c} (SR) seeded companion}
\label{sec:sr-global}
Stefanica and Radoi\v{c}i\'{c}~\cite{sr2017} derived a closed-form global approximation based on the
P\'{o}lya bound for~$\Phi$.  It handles all moneyness regimes via
careful case analysis. Our implementation uses a compact guarded
form of the SR approximation, with Taylor-limit branches at singular
boundaries.

StrontiumIV is an implied volatility solver based on the 
SR seed plus Householder iterations on the log-price, or $\ln(c)$, objective used
by Healy~\cite{healy2024}, made robust by the same normalisation, microscopic
Bachelier guard, and high-price complementary guard used here.  It removes the
Li/asymptotic dispatch and uses the SR approximation~\cite{sr2017} as the
universal non-Bachelier seed before the same fixed H3 refiner used by FlashIV.
The trade-off is direct: the guarded SR seed costs 26.2\,ns on the CLY-3D case
mix, against the 4.8\,ns average of the Li/asymptotic hybrid.  In
Table~\ref{tab:accuracy}, StrontiumIV runs at 179\,ns on CLY-3D versus 137\,ns
for FlashIV.

\subsubsection{Accuracy of the raw initial guesses}
\label{sec:guess-accuracy}

Before any iterative refinement, the three seeds have quite
different accuracy profiles.  Table~\ref{tab:guessaccuracy} reports
relative error in total volatility,
$|v_0-v_{\rm ref}|/v_{\rm ref}$, over the union of the eight benchmark
datasets.  The benchmark uses the same dataset construction as
Table~\ref{tab:accuracy}.

\begin{table}[!htbp]
\caption{Pure initial-guess accuracy before any H3/cubic refinement.
Errors are relative errors in total volatility over all \num{68551}
benchmark cases.  The cost column is the CLY-3D micro-benchmark cost from
Table~\ref{tab:costs}.\label{tab:guessaccuracy}}
\begin{tabularx}{\textwidth}{lCCCCC}
\toprule
\textbf{Seed} & \textbf{Cost (ns)} & \textbf{Mean} & \textbf{Median}
  & \textbf{95th pct.} & \textbf{Max} \\
\midrule
FlashIV Li/asym       & 4.8  & $1.38{\times}10^{-1}$ & $1.25{\times}10^{-1}$
  & $2.84{\times}10^{-1}$ & $3.05{\times}10^{1}$ \\
StrontiumIV SR        & 26.2 & $3.35{\times}10^{-2}$ & $3.52{\times}10^{-2}$
  & $5.75{\times}10^{-2}$ & $1.21{\times}10^{-1}$ \\
ThiopheneIV Choi L3   & 17.7 & $1.57{\times}10^{-1}$ & $1.60{\times}10^{-1}$
  & $2.82{\times}10^{-1}$ & $3.57{\times}10^{-1}$ \\
\bottomrule
\end{tabularx}
\end{table}

The SR seed is therefore the most accurate raw approximation in this
benchmark: even its worst relative total-volatility error is about
12\%.  FlashIV's seed is deliberately cheaper and rougher:
its large maximum comes from a handful of extreme Corners and Stress
cases where the asymptotic seed is far from the true low total
volatility.  Those cases are precisely why the solver is assessed after
the solver-specific log-price refinement rather than by seed accuracy alone.

\subsection{Eliminating the Convergence Loop}
\label{sec:fixed}

Traditional Householder solvers iterate until a convergence criterion
is met, checking residuals at every step.  This introduces overhead
from branch misprediction (the loop exit is data-dependent), the
convergence check itself, and lost opportunities for compiler
optimisation (the loop body cannot be fully unrolled or scheduled
across iterations).

We measured this overhead at $\sim$19\,ns by comparing two otherwise matched
prototype paths, one with a convergence loop after the fast pre-step
(160\,ns) and one with fixed iterations (141\,ns), on the CLY-3D dataset.

Instead of looping, we execute exactly \textbf{two} unconditional
exact H3 steps.  A \textbf{conditional third step} fires only when the
log-price residual on entry to the second exact step is still large:
\begin{linenomath}
\begin{equation}\label{eq:safety}
  \text{if } |f(v_1)| \ge 10^{-4}: \quad\text{execute 3rd H3 step after } v_2.
\end{equation}
\end{linenomath}

\begin{Proposition}
On the CLY-3D benchmark (\num{51321} cases), the conditional step fires
for 11 cases (0.02\%), exclusively in the deepest OTM zones
($|x| > 2$), adding $< 0.01$\,ns in amortised cost.
\end{Proposition}

\begin{Remark}[Why two steps suffice]
The pre-step (with the fast erfcx) promotes the initial guess from
$\sim$2--3 digits to $\sim$6--8 digits.  The first exact H3 step
(quartic convergence) then yields $\sim$24--32 digits, well beyond
double precision.  The second exact step drives the residual to the
limits of IEEE~754 arithmetic ($\sim 10^{-16}$).  Only in pathological
cases where the pre-step achieves fewer than 4 digits (very deep OTM
with an imprecise asymptotic guess) is a third step needed.
\end{Remark}

Figure~\ref{fig:fixed-x-convergence} visualises this fixed-count
compression on fixed log-moneyness slices, in the style of J\"ackel's
diagnostic plots.  The raw seed errors are broad, the fast pre-step moves
the slices into the local H3 basin, and the exact steps collapse the error
to the double-precision floor.

\begin{figure}[!htbp]
\centering
\includegraphics[width=\textwidth]{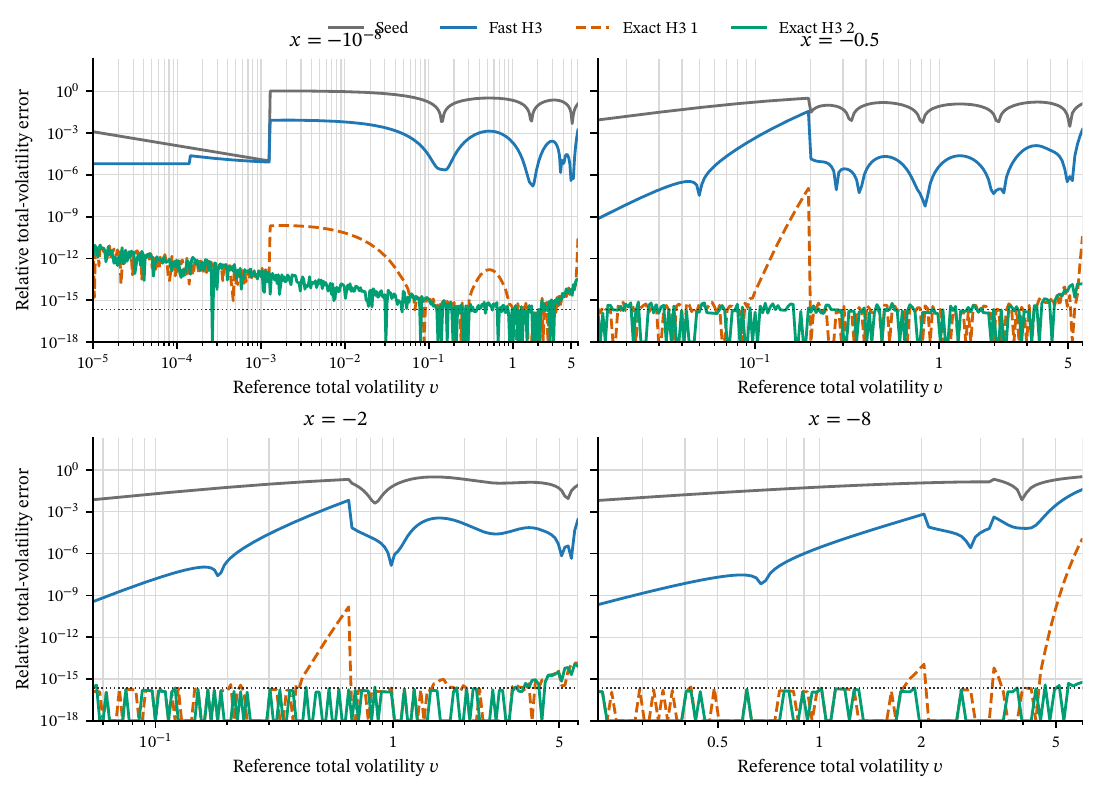}
\caption{Fixed-$x$ convergence slices for the FlashIV H3 chain.  Each
panel fixes log-moneyness and varies the reference total volatility.  The
curves show relative total-volatility error after the seed, the fast H3
pre-step, and the first two exact H3 steps; values are floored at
$10^{-18}$ only for log-scale plotting.\label{fig:fixed-x-convergence}}
\end{figure}

\subsection{Algorithm Summary}
\label{sec:algorithm-summary}

Algorithm~\ref{alg:hotpath} summarises the FlashIV hot path.  The companion
StrontiumIV and ThiopheneIV rows in the experiments share the same
normalisation, microscopic-price guard, and high-price complementary guard; the shared
guard derivations are in the ThiopheneIV paper~\cite[Appendices~A.3 and~A.4]{lefloc2026thiophene},
while the seed/refinement internals differ by solver.

\begin{algorithm}[H]
\caption{FlashIV hot path.\label{alg:hotpath}}
\begin{algorithmic}[1]
\Require $x\le0$, $e^x$, normalised OTM price $0<c<1$, $\ln\left(c\right)$, maturity $T$
\If{$c\le 10^{-6}$ and $|x|\le 10^{-8}$}
  \State \Return Bachelier-limit volatility (Mills-ratio tail, Bachelier rational seed + expansion polish, or defensive zero)
\EndIf
\State $v_0 \gets$ Li/asymptotic dispatch of Section~\ref{sec:guess}
\State $v_0 \gets \max(v_0,10^{-10})$
\If{$c\ge0.99$}
  \State $v_0 \gets \max(v_0,\text{upper-price asymptotic seed})$
  \State \Return three Halley steps on the complementary $\ln(1-c)$ objective
\EndIf
\State $v_1 \gets$ one H3 step using the fast A\&S erfcx
\State $v_2 \gets$ one exact H3 step using full-precision erfcx
\State $r_2 \gets$ log-price residual at $v_2$
\State $v_3 \gets$ one exact H3 step using full-precision erfcx
\If{$|r_2|\ge10^{-4}$}
  \State $v_3 \gets$ one additional exact H3 safety step
\EndIf
\State \Return $v_3/\sqrt{T}$
\end{algorithmic}
\end{algorithm}

\subsection{Scalar Inlining of Derivative Evaluations}
\label{sec:scalar}

This subsection documents an implementation detail that we separate from the
algorithmic contributions.

The direct implementation of the Householder iteration computes the
objective and its three derivatives in a helper function that returns
them as a four-element container.  Replacing that container return with scalar
temporaries removes the representation overhead without changing the iteration.
In the benchmark, this reduces per-case time from 210\,ns to 171\,ns
($-19$\%).  Since this is a representation choice rather than a solver design
choice, the algorithmic comparisons below use the scalar-inlined baseline.

%%%%%%%%%%%%%%%%%%%%%%%%%%%%%%%%%%%%%%%%%%
\FloatBarrier
\section{Experimental Setup}
\label{sec:setup}

\subsection{Reference Price Computation}
\label{sec:refprice}
The textbook Black--Scholes formula
$c = \Phi\left(d_1\right) - e^{-x} \Phi\left(d_2\right)$ suffers from catastrophic
cancellation for deep OTM options, and even the erfcx/log objective used by the
unpolished solver can lose relative bits in nearly-ATM, very-low-volatility
cases because it subtracts two close erfcx values \cite{lefloc2026thiophene}.  
The main accuracy table uses saved OTM-normalised prices generated from a
multiprecision Black evaluator at the known reference volatility and then
rounded to double precision.  
%The fixed-$x$ comparison, polish diagnostic, and Pareto view instead use J\"ackel's cancellation-avoiding normalised-Black price, i.e. the expanded small-volatility evaluation used by \emph{Let's Be Rational}. We call this quantity the expanded J\"ackel reference price below; it is the quantity that the optional final Newton correction matches.

\subsection{Test Datasets}
\label{sec:datasets}

To ensure robustness, we benchmark on eight datasets spanning the
full range of practically relevant option parameters
(Table~\ref{tab:datasets}).

\begin{table}[!htbp]
\caption{Test datasets.\label{tab:datasets}}
\begin{tabularx}{\textwidth}{lrX}
\toprule
\textbf{Dataset} & \textbf{Cases} & \textbf{Description} \\
\midrule
CLY-3D    & \num{51321} & Cui--Liu--Yao three-dimensional grid:
            $K \in [105,800]$, $T \in [0.01,2]$,
            $\sigma \in [0.01,0.99]$, filtered at price $10^{-20}$ \\
CLY-20    & \num{1600} & Cui--Liu--Yao fixed-volatility surface:
            $\sigma=20\%$, $K \in [105,180]$, $T \in [0.1,2]$ \\
CLY-80    & \num{1600} & Cui--Liu--Yao fixed-volatility surface:
            $\sigma=80\%$, $K \in [105,800]$, $T \in [0.1,2]$ \\
J\"ackel  & \num{5182}  & Wide moneyness: $K/F \in [0.5, 8]$,
            $\sigma$ up to 4.0 \\
Market    & \num{7151}  & Realistic: $K/F \in [0.7, 1.5]$,
            $T$ from 1/252 to 5\,yr \\
Corners   & 278         & Edge cases: low-vol/short-mat,
            high-vol/deep-OTM, near-ATM small-price, and high-price near-bound cases \\
Stress    & \num{1270}  & Extremes: $K/F$ up to $100\times$,
            $T \in [0.001, 10]$ \\
HighVol   & 149         & Fallback stress zone: $|x| \ge 3$,
            $c \in (0.05, 0.95)$, $\sigma$ up to 2.5 \\
\bottomrule
\end{tabularx}
\end{table}

The first three datasets reproduce the comparison grids from
Cui--Liu--Yao~\cite{cui2025}: the main three-dimensional grid and the
two fixed-volatility surface grids used in their numerical comparison
with previous literature.  The HighVol dataset
specifically targets the region $|x| \ge 3$ with large option prices,
where Li's rational approximation is outside its fitted domain and
the asymptotic OTM guard only accepts part of the set; cases are
filtered to $c < 0.95$ to reflect the practical use of put--call
parity for deep ITM options.
The accompanying source distribution gives the exact construction rules used by
the benchmark generator.

\subsection{Solvers Under Comparison}
\label{sec:solvers}

The main comparison uses only production-facing rows.  FlashIV and FlashIV+
are the methods studied here; J\"ackel is the reference solver; StrontiumIV,
ThiopheneIV, and ThiopheneIV+ are compact comparison rows that contextualise
alternative seed and polish choices.  Intermediate development variants are not
part of the main comparison because the anatomy section below already explains
the performance components.
All micro-benchmark timings use the same
harness: minimum of 500 sweeps and 3 independent runs in the implementation.

\begin{itemize}
  \item \textbf{J\"ackel's \emph{Let's Be Rational}~\cite{jaeckel2017}.}
        Region-dependent asymptotic expansions, log-space iteration,
      complementary objective.  The comparison uses the normalised API
        with $\beta=c\sqrt{e^x}$ and the original two-iteration default.
        243\,ns on CLY-3D.

  \item \textbf{FlashIV (this paper).}
        Li/asymptotic seed dispatch, one fast H3 pre-step, two exact H3 steps,
        and a rare conditional safety step.  137\,ns on CLY-3D.

  \item \textbf{FlashIV+.}
        FlashIV with one final J\"ackel--Newton correction against the expanded
    J\"ackel reference price, except under the upper polish guard $c\ge0.99$,
    where the complementary gap branch is retained.  181\,ns on CLY-3D.

  \item \textbf{StrontiumIV.}
        Guarded SR approximation~\cite{sr2017} for all non-Bachelier
        inputs, with no Li coefficients and no Li/asymptotic guess-domain
      switching, followed by Healy's SR--Householder-on-log-price idea
      with the same robust guard layer and fixed H3 polish as FlashIV.
      179\,ns on CLY-3D.

  \item \textbf{ThiopheneIV.}
        Choi--Huh--Su L3 lower-bound seed~\cite{choi2023} for all
      non-Bachelier inputs, followed by lower-tail Euler--Chebyshev or
      upper-tail Halley steps and the same production guards.  166\,ns on CLY-3D.

  \item \textbf{ThiopheneIV+.}
        ThiopheneIV with one final J\"ackel--Newton correction against the
      expanded J\"ackel reference price on the lower-price half $c\le1/2$.
      214\,ns on CLY-3D.
\end{itemize}
%%%%%%%%%%%%%%%%%%%%%%%%%%%%%%%%%%%%%%%%%%
\FloatBarrier
\section{Results}
\label{sec:results}

\subsection{Accuracy and Timing}

Table~\ref{tab:accuracy} reports the main accuracy and latency comparison on
the eight benchmark datasets.  Input prices for the accuracy block are generated
from a multiprecision Black price at the known total volatility
$v_{\rm ref}$, rounded to double precision, and then passed to the solvers as OTM-normalised prices.  Errors are
measured in ulps of the reference total volatility: for a solver output
$\hat v=\hat\sigma\sqrt{T}$, the per-case error is
\[
  \frac{|\hat v-v_{\rm ref}|}{\operatorname{nextUp}(v_{\rm ref})-v_{\rm ref}}.
\]
Thus one ulp is the local double-precision spacing at $v_{\rm ref}$; at
$v_{\rm ref}=0.20$, a 7-ulp error is about $1.94\times10^{-16}$ in total
volatility.  The table is best read through the FlashIV, FlashIV+, and J\"ackel
rows: FlashIV is the low-latency row, FlashIV+ adds one Newton step below
the upper polish guard $c<0.99$, and J\"ackel is the reference
solver.  StrontiumIV, ThiopheneIV, and ThiopheneIV+ provide compact context for
alternative seed/refinement choices.

\begin{table}[!htbp]
\caption{Accuracy against rounded multiprecision Black reference prices and latency
by dataset.\label{tab:accuracy}}
\scriptsize
\setlength{\tabcolsep}{4pt}
\begin{tabularx}{\textwidth}{lCCCCCCCC}
\toprule
& \textbf{CLY-3D} & \textbf{CLY-20} & \textbf{CLY-80}
& \textbf{J\"ackel} & \textbf{Market} & \textbf{Corners}
& \textbf{Stress} & \textbf{HighVol} \\
\midrule
\multicolumn{9}{l}{\textbf{Accuracy -- max error (ulp of reference total volatility)}} \\[2pt]
J\"ackel      &  23 &  5 & 4 &  25 &  29 & 240 &  33 & 2 \\
FlashIV       & 130 & 13 & 4 &  93 & 304 & 329 & 208 & 7 \\
FlashIV+      &  23 &  5 & 5 &  13 &  29 &  41 &  33 & 3 \\
StrontiumIV   & 116 & 19 & 4 & 106 & 117 & 329 & 321 & 9 \\
ThiopheneIV   & 133 & 62 & 7 &  89 & 177 & 329 & 138 & 2 \\
ThiopheneIV+  &  24 &  5 & 5 &  13 &  29 &  41 &  33 & 2 \\[4pt]
\multicolumn{9}{l}{\textbf{FlashIV+ accuracy -- max absolute total-volatility error}} \\[2pt]
FlashIV+
  & $6.7{\times}10^{-16}$ & $1.1{\times}10^{-16}$
  & $6.7{\times}10^{-16}$ & $1.2{\times}10^{-14}$
  & $1.8{\times}10^{-15}$ & $2.7{\times}10^{-15}$
  & $8.9{\times}10^{-16}$ & $2.7{\times}10^{-15}$ \\[4pt]
\multicolumn{9}{l}{\textbf{Latency (ns/call)}} \\[2pt]
J\"ackel      & 243 & 225 & 225 & 224 & 204 & 261 & 232 & 203 \\
FlashIV       & 137 & 137 & 136 & 144 & 145 & 148 & 147 & 142 \\
FlashIV+      & 181 & 181 & 184 & 199 & 194 & 202 & 199 & 200 \\
StrontiumIV   & 179 & 175 & 174 & 181 & 183 & 178 & 179 & 186 \\
ThiopheneIV   & 166 & 164 & 161 & 171 & 171 & 173 & 175 & 169 \\
ThiopheneIV+  & 209 & 209 & 211 & 209 & 216 & 222 & 224 & 207 \\
\bottomrule
\end{tabularx}
\end{table}

The optional final Newton correction changes the trade-off rather than the
solver architecture.  On the CLY-3D grid, FlashIV moves from 130 to 23
total-volatility ulps while remaining faster than J\"ackel's unpolished
comparison path.  On the J\"ackel grid, the high-price guard is decisive:
FlashIV+ reports 13 ulps because cases with $c\ge0.99$ use the complementary
$\ln(1-c)$ objective instead of a direct beta-price Newton polish.  The same
guard reduces the base FlashIV and StrontiumIV J\"ackel maxima to 93 and
106 ulps, respectively.  The Corners dataset is the most demanding because it includes
cancellation-sensitive near-ATM small-price cases and high-price near-bound
cases; the FlashIV+ absolute-error row shows that the 41-ulp maximum there corresponds to a
$2.7\times10^{-15}$ total-volatility error.  ThiopheneIV implementation details
and the analogous lower-Chebyshev / upper-Halley comparison are given in
\cite[Sections~3 and~5]{lefloc2026thiophene}.

Figure~\ref{fig:jaeckel-comparison-slice} gives a fixed-log-moneyness
slice check.  Prices below $10^{-300}$ are excluded, matching the
benchmark-generation floor used in the test sets.

\begin{figure}[!htbp]
\centering
\includegraphics[width=\textwidth]{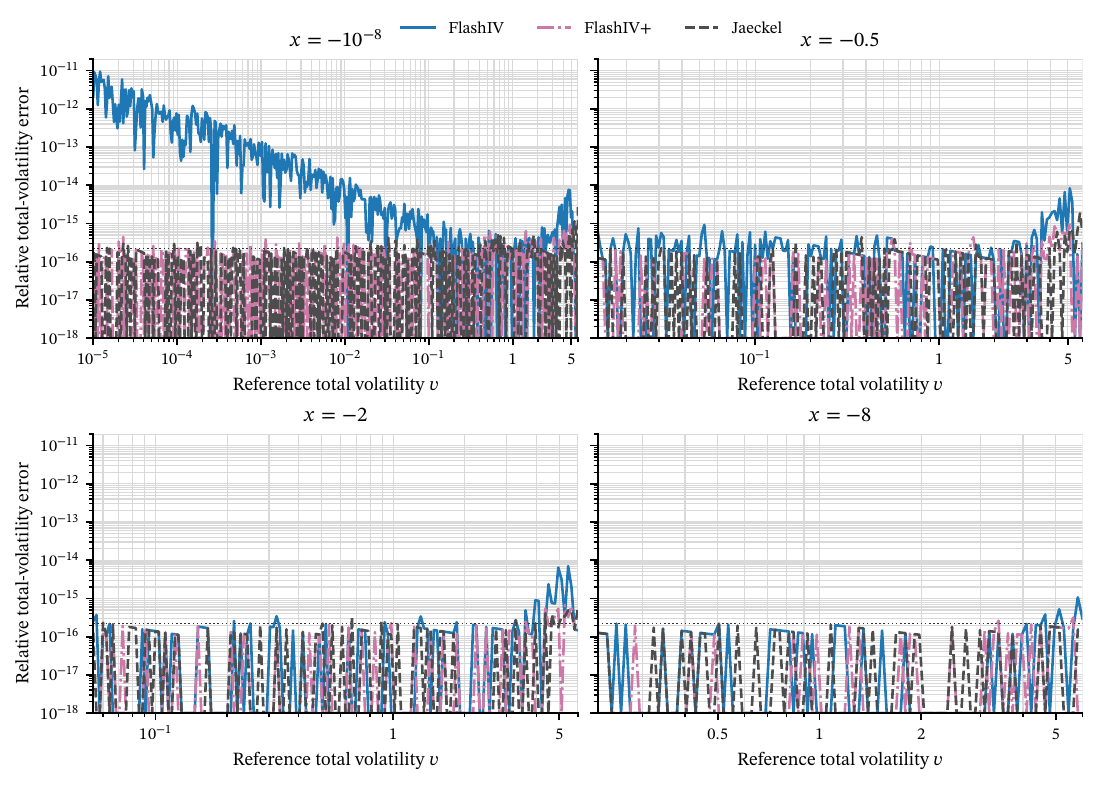}
\caption{Fixed-$x$ comparison of FlashIV, FlashIV+, and J\"ackel's normalised
solver using expanded J\"ackel reference prices.
Each panel varies the reference total volatility at fixed log-moneyness
and reports relative total-volatility error against the known
$v_{\rm ref}$.  FlashIV+ collapses most of the remaining discrepancy against the
expanded J\"ackel reference price while preserving the fixed-count FlashIV path
plus a guarded final Newton correction.\label{fig:jaeckel-comparison-slice}}
\end{figure}

The larger FlashIV error in the top-left panel
($x=-10^{-8}$, $v\lesssim10^{-3}$) is a pricing-formula limit, not an H3
convergence failure.  The input prices are generated by J\"ackel's expanded
reference price, whose small-volatility branch uses a
Taylor-expanded evaluation to avoid cancellation.  Unpolished FlashIV solves
the erfcx/log objective instead; in this near-ATM, very-small-total-variance
corner the two erfcx terms are nearly equal, so the erfcx formula loses enough
relative bits that its root is slightly displaced from the expanded J\"ackel
reference root.  FlashIV+ removes this discrepancy by applying one final Newton
step on the expanded J\"ackel reference price, which is why it returns to the
double-precision floor in that panel.

Figure~\ref{fig:latency-accuracy-pareto} combines the selected latency rows of
Table~\ref{tab:accuracy} with the multiprecision-reference accuracy
recomputation.
The
small pale points show the non-CLY datasets, while CLY-3D, CLY-20, and
CLY-80 are highlighted.  The solid black step line is the nondominated
frontier over each solver's worst CLY observation; in this benchmark the
worst CLY error for every solver occurs on CLY-3D, so the frontier lies
on a subset of the CLY-3D markers.  Solver labels abbreviate each solver's
optional guarded expanded J\"ackel price polish by a trailing `+'.

\begin{figure}[!htbp]
\centering
\includegraphics[width=0.92\textwidth]{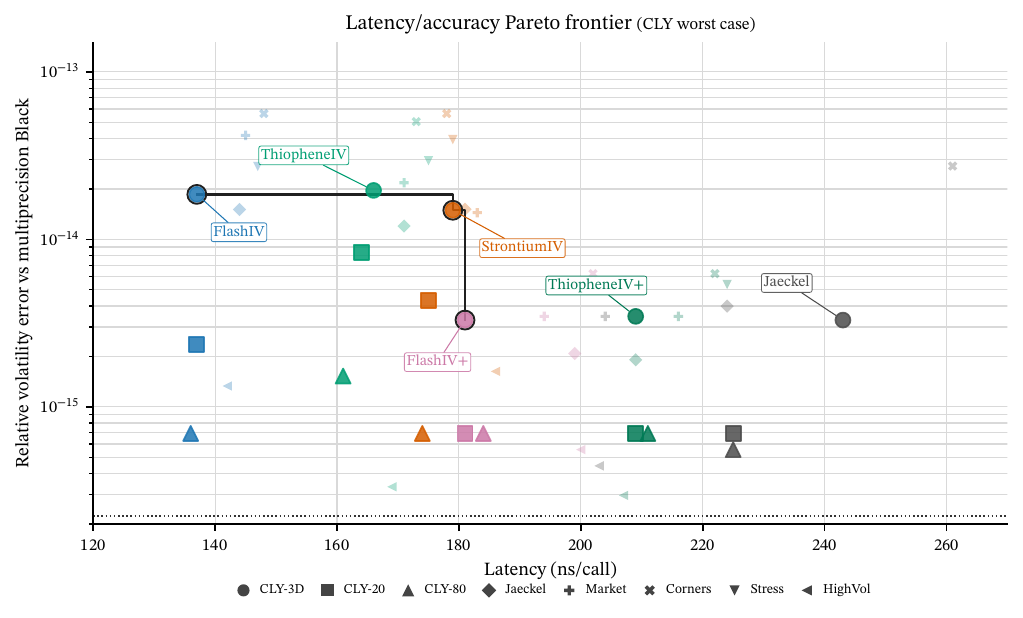}
\caption{Latency/accuracy Pareto frontier across the eight benchmark datasets
with the three CLY datasets highlighted.
The accuracy coordinate is recomputed from prices generated with the
multiprecision Black reference.  The solid step frontier is computed from
each solver's worst error over CLY-3D, CLY-20, and CLY-80; here those
worst CLY points are all CLY-3D observations.  The `+' suffix denotes each
solver's optional guarded expanded J\"ackel price polish.  FlashIV anchors the low-latency
end of the CLY frontier, while the polished rows mark the
near-machine-precision region.\label{fig:latency-accuracy-pareto}}
\end{figure}

Table~\ref{tab:accuracy} compares final solver configurations.  On CLY-3D,
FlashIV runs at 137\,ns, about 44\% lower latency than J\"ackel's solver (243\,ns).
The Anatomy section explains where that latency is spent.

The FlashIV+ row in Table~\ref{tab:accuracy} uses the optional expanded
J\"ackel price polish only below the upper polish guard $c\ge0.99$; in the
guarded high-price region it keeps the complementary gap output. 

\subsection{Anatomy of FlashIV's 137 Nanoseconds}
\label{sec:anatomy}

Table~\ref{tab:costs} breaks down the cost of individual operations.

\begin{table}[!htbp]
\caption{Micro-benchmarked operation costs in the implementation.\label{tab:costs}}
\begin{tabularx}{\textwidth}{lC}
\toprule
\textbf{Operation} & \textbf{Cost (ns)} \\
\midrule
Li rational guess                   & 3.3 \\
Asymptotic OTM guess                & 3.4 \\
Li/asym guess combined (average)    & 4.8 \\
SR guess                            & 26.2 \\
Thiophene L3 guess                  & 17.7 \\
1 fast erfcx (A\&S 7.1.26)          & 2.3 \\
1 exact erfcx (Boost/Commons)       & 5.2 \\
1 fast H3 step                      & 21.8 \\
1 exact H3 step                     & 31.6 \\
Natural logarithm                   & 4.7 \\
Exponential                         & 3.6 \\
Square root                         & 0.7 \\
\addlinespace
Optional J\"ackel--Newton polish (incremental) & about 45--55 \\
\bottomrule
\end{tabularx}
\end{table}

The optional J\"ackel--Newton row is a composite incremental cost, not a
primitive operation: it includes one expanded J\"ackel price evaluation,
the analytic normalised vega, and the Newton correction arithmetic.  It is
therefore excluded from the default FlashIV cost model below, but it
explains the opt-in timings reported in Table~\ref{tab:accuracy}.

The total cost assembles as:
\begin{linenomath}
\begin{equation}\label{eq:costmodel}
  \underbrace{4.8}_{\text{guess}} +
  \underbrace{21.8}_{\text{fast H3}} +
  \underbrace{2 \times 31.6}_{\text{exact H3}} +
  \underbrace{47}_{\text{overhead}} \approx 137\,\text{ns},
\end{equation}
\end{linenomath}
where the roughly 47\,ns overhead comprises $\ln\left(c_{\mathrm{obs}}\right)$
computation ($\sim$4.7\,ns), input normalisation, three logarithms
inside the H3 steps ($\sim$14\,ns), and
loop, branch, and function-call overhead.

Equation~\eqref{eq:costmodel} gives the empirical decomposition behind the
137\,ns CLY-3D timing.  The largest implementation-level contributor is scalar
inlining, which removes the container-return overhead of the original derivative
interface.  The algorithmic contributors are the fast pre-step, the
fixed-count H3 path, and the Li/asymptotic seed dispatch, which together make
the common path predictable and keep the expensive guarded SR seed off most OTM
cases.

%%%%%%%%%%%%%%%%%%%%%%%%%%%%%%%%%%%%%%%%%%
\FloatBarrier
\section{Discussion}
\label{sec:discussion}

The experiments show that FlashIV's latency gain is architectural rather than a
single formula substitution.  The common path uses a 4.8\,ns average
Li/asymptotic seed, a 21.8\,ns fast-erfcx H3 pre-step, and two exact H3 steps
with no adaptive convergence loop.  The result is mostly straight-line scalar
arithmetic, with boundary branches moved out of the common path.

J\"ackel's \emph{Let's Be Rational}~\cite{jaeckel2017} remains the appropriate
reference point.  It is widely used, carefully engineered, and designed for the
full implied-volatility domain.  In the normalised comparison used here,
J\"ackel takes 243\,ns on CLY-3D.  FlashIV is not presented as a more accurate
solver than J\"ackel; it reaches the same practical accuracy class with lower
latency in this benchmark.

StrontiumIV shows that the
guarded SR seed is a strong universal seed, but the 26.2\,ns seed cost raises
latency.  ThiopheneIV explores a different design point: it starts from the
Choi--Huh--Su L3 lower-bound seed and applies lower-tail Euler--Chebyshev and
upper-tail Halley steps, which support monotone-convergence arguments but cost
more than FlashIV's cheaper Li/asymptotic dispatch.  Its construction and proof are described in the ThiopheneIV
paper~\cite[Section~3 and Appendix~D]{lefloc2026thiophene}.  In the present
benchmark it is faster than J\"ackel but slower than FlashIV because it spends
more on the seed/refinement stage.

The guards are as important as the iteration.  Normalised OTM pricing avoids
intrinsic-value cancellation; the erfcx/log objective keeps deep tails
well-scaled; the microscopic Bachelier branch handles prices for which the
Black tail difference is no longer a useful floating-point object; and the
complementary upper-price branch handles high prices near the upper bound.
These branches are narrow, but removing them would make the fast path less robust.  The
ThiopheneIV paper~\cite[Appendices~A.3 and A.4]{lefloc2026thiophene} gives the
fuller corner-case derivation; FlashIV reuses the same production guard layer
around its Li/asymptotic seed and H3 refiner.

The absolute nanosecond values depend on the runtime, the erfcx implementation, and
single-scalar execution.  Faster special functions or vectorised batches would
change the totals.  The more portable conclusion is that cheap
distribution-aware seeds, fixed high-order refinement, and carefully isolated
guard branches can reduce latency without leaving the J\"ackel accuracy class.

%%%%%%%%%%%%%%%%%%%%%%%%%%%%%%%%%%%%%%%%%%
\FloatBarrier
\section{Conclusions}
\label{sec:conclusion}

This paper presented FlashIV, a fixed-count solver for Black--Scholes implied
volatility based on normalised OTM prices, an erfcx/log-price objective, a cheap
distribution-aware initial guess, one fast H3 pre-step, and two exact H3
refinement steps.  On \num{68551} benchmark cases it reaches near-machine-precision
accuracy, with maximum relative volatility error below $7.3\times10^{-14}$, and
runs in 137\,ns per call on the CLY-3D dataset in the reported benchmark.

FlashIV is the fastest solver reported here: 137\,ns on
CLY-3D versus 166\,ns for ThiopheneIV, 179\,ns for StrontiumIV, and 243\,ns for
the J\"ackel \emph{Let's Be Rational} reference
implementation.  The speedup is cumulative: scalar derivative inlining removes
a container-return overhead, while the algorithmic gains come from the fast erfcx
pre-step, fixed iteration count, and cheap seed selection.

The optional guarded J\"ackel--Newton polish changes the trade-off rather than
the main conclusion.  It reduces discrepancies against the expanded J\"ackel
reference price, but raises FlashIV latency to 181\,ns on CLY-3D and about
193\,ns on average across the eight datasets.  Whether that last correction is
worthwhile depends on whether an application needs agreement with this reference
price to the last few ulps.  Ideally, production systems would use the same
high-accuracy J\"ackel-style Black price for pricing, calibration, and inversion.
In practice, the straightforward Black--Scholes formula built from a library
normal CDF is extremely common; in such systems the pricing formula itself is
usually less accurate than the expanded J\"ackel reference price, so the
polishing step may provide no practical benefit.

StrontiumIV and ThiopheneIV support the same architectural conclusion from two
other seed choices.  StrontiumIV makes Healy's SR--Householder-on-log-price approach robust and
accurate over the full benchmark domain, but it is slower than ThiopheneIV
because the guarded SR seed and its case handling are more expensive than the
Choi L3 seed used by ThiopheneIV before the shared guard/refinement layers.
ThiopheneIV's strong advantage is a proof-backed convergence 
while FlashIV is the low-latency member of the family,
using the cheaper distribution-aware seed while retaining the guard structure
needed for robust production use.

%%%%%%%%%%%%%%%%%%%%%%%%%%%%%%%%%%%%%%%%%%
\vspace{6pt}

\authorcontributions{F.L.F. conceived the algorithm. J.H. implemented the code and ran the experiments. F.L.F. wrote the paper.}

\funding{This research received no external funding.}

\dataavailability{The source code used to generate the benchmark
tables, including the solver variants, benchmark harnesses, and dataset
builders, is available in the accompanying source distribution, including the
exact dataset grids used for the tables.}

\acknowledgments{The author thanks Gary Kennedy for a thorough review and
excellent feedback.}

\conflictsofinterest{The author declares no conflicts of interest.}

%%%%%%%%%%%%%%%%%%%%%%%%%%%%%%%%%%%%%%%%%%
\begin{adjustwidth}{-\extralength}{0cm}

\reftitle{References}

\end{adjustwidth}

\end{document}